\documentclass[conference]{IEEEtran}
\usepackage{bbding}
\usepackage[T1]{fontenc}
\usepackage{aurical}
\usepackage{huncial}
\usepackage{linearb}
\usepackage{textcomp}
\usepackage{amsmath}
\usepackage{amssymb}
\usepackage[amsmath,amsthm,thmmarks]{ntheorem}
\usepackage{amsfonts}
\usepackage{amssymb,epic,eepic}
\usepackage{amsmath}
\usepackage{stmaryrd}
\usepackage{capt-of}
\usepackage{graphicx}
\usepackage[usenames,dvipsnames]{pstricks}
\usepackage{epsfig}
\usepackage{pst-grad}
 \usepackage{pst-plot}
\usepackage{mathtools}
\usepackage{textfit} 
\usepackage{subfigure}
\usepackage{subcaption}
\usepackage{float}

\newtheorem{theorem}{Theorem}[section]

\newtheorem{corollary}[theorem]{Corollary}
\newtheorem{definition}[theorem]{Definition}

\theoremstyle{definition}

\theoremstyle{remark}

\onecolumn

\ifCLASSINFOpdf
\else
\fi
\IEEEoverridecommandlockouts
\begin{document}
%
\title{Quantum Byzantine Multiple Access Channels}

\author{\IEEEauthorblockN{Minglai Cai and
Christian Deppe 
}
\IEEEauthorblockA{\fontsize{10}{11}\selectfont \textit{
Institut f\"ur Nachrichtentechnik,
Technische Universit\"at Braunschweig
Germany,}\allowdisplaybreaks\\ 
minglai.cai@tu-braunschweig.de,  christian.deppe@tu-braunschweig.de }

}


%


\maketitle

\IEEEpeerreviewmaketitle

\begin{abstract}
In communication theory, attacks like eavesdropping or jamming are typically assumed to occur at the channel level, while communication parties are expected to follow established protocols. But what happens if one of the parties turns malicious? In this work, we investigate a compelling scenario: a multiple-access channel with two transmitters and one receiver, where one transmitter deviates from the protocol and acts dishonestly. To address this challenge, we introduce the Byzantine multiple-access classical-quantum channel and derive the communication region of rate for this adversarial setting.
\end{abstract}

\section{Introduction}
The rapid advancements in modern communication systems open up new possibilities while simultaneously exposing us to numerous cybersecurity threats. Many of these systems remain imperfect and vulnerable to hostile attacks, necessitating robust protection against malicious interference. 

A prominent example is the forthcoming 6G communication standard, which will integrate key technological infrastructure for the metaverse and other digital-twin applications. Given its profound impact on future applications, particularly in sensitive aspects of human well-being, ensuring a high degree of trustworthiness is imperative (\cite{Bo/Bo/Mo/fi,Fitzek2021,Fitzek2022,Fettweis2022}).

Communication models that account for a jammer attempting to disrupt proper communication between legitimate parties-often represented by the arbitrarily varying channel-have been extensively studied in both classical (\cite{Bl/Br/Th2}, \cite{Ahl0}) and quantum information theory (\cite{Ahl/Bli}, \cite{Bj/Bo/Ja/No}, \cite{Ahl/Bj/Bo/No}).

A multiple-access channel is a communication channel where two or more senders transmit information to a common receiver. The optimal coding strategies and capacity regions for classical multiple-access channels have been established in \cite{Ahl4} and \cite{Lia}, while those for multiple-access quantum channels have been derived in \cite{Win2} and \cite{Ya/Ha/De2}.

A multiple-access channel with a jammer of known identity is referred to as an arbitrarily varying multiple-access channel. The capacity regions of classical arbitrarily varying multiple-access channels were determined in \cite{Ja}. The symmetrizability conditions for these channels were studied in \cite{Gu} and \cite{Ahl/Ca}. In the classical-quantum setting, \cite{Da/Wa/Ha} established the capacity region by introducing a novel coding technique that does not rely on conditional typical subspace projectors.

However, modern digital communication systems are increasingly vulnerable to cyber threats, not only from external attackers such as hackers or government agencies but also from within. In some cases, protection against external threats alone is insufficient, as certain parties may act dishonestly or even attempt to disrupt the communication of honest participants.

Unlike arbitrarily varying channel models, the Byzantine model considers scenarios where one of the legitimate users may behave adversarially. Byzantine attacks on network coding have been extensively studied in \cite{Ko/To/Da}, \cite{Ja/La/Ka/Ho/Ka/Me}, and \cite{Wa/Si/Ks}.

\begin{center}\begin{figure}[H]
\begin{minipage}{0.32\linewidth}
        \centering
\includegraphics[width=0.95\linewidth]{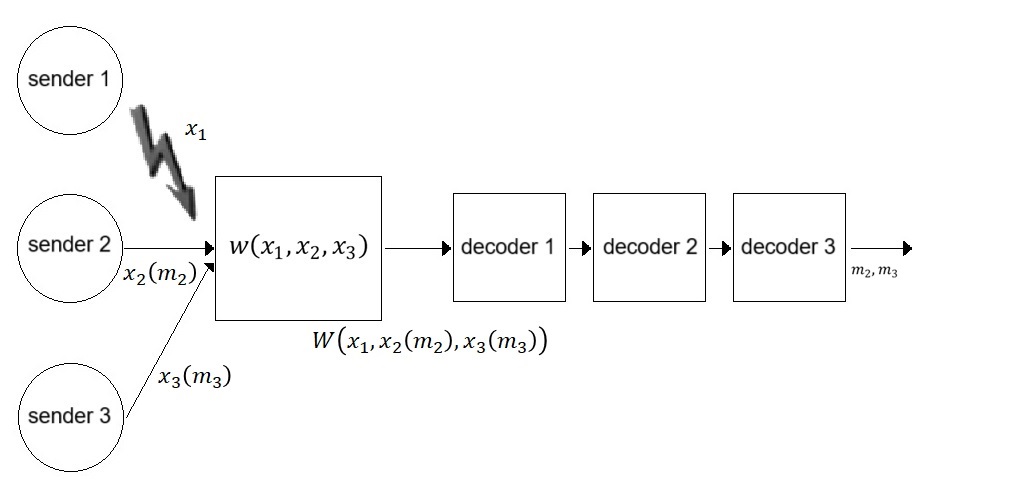}
 \end{minipage}
\begin{minipage}{0.32\linewidth}
        \centering
\includegraphics[width=0.95\linewidth]{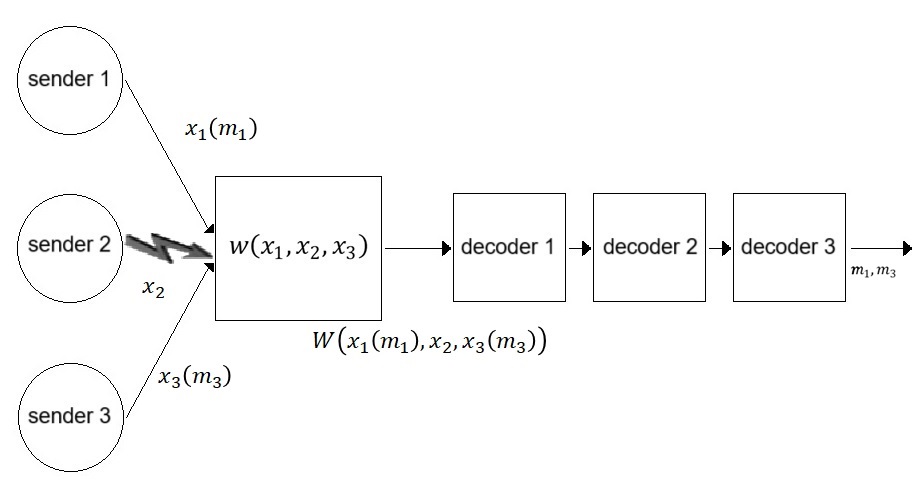}
 \end{minipage}
\begin{minipage}{0.32\linewidth}
        \centering
\includegraphics[width=0.95\linewidth]{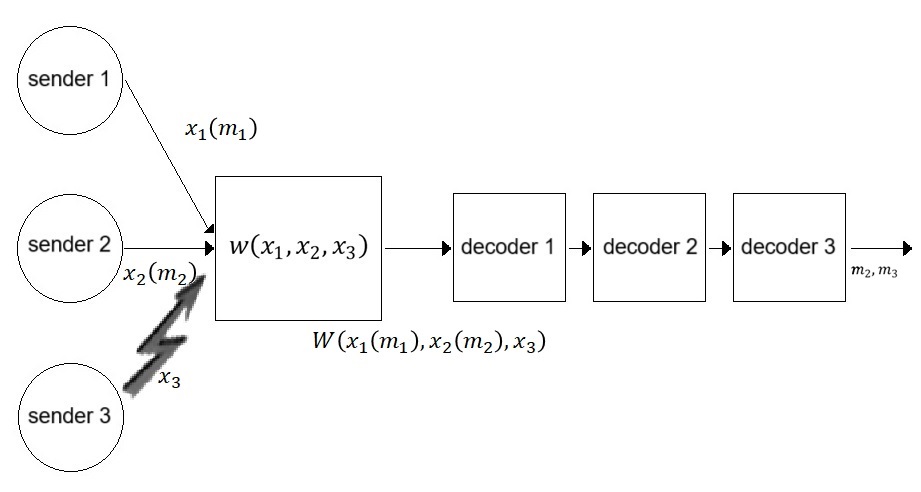}
 \end{minipage} 
 \caption{classical-quantum 3-user byzantine-multiple access channel, one user may be adversarial.}
\label{3userbymac}
\end{figure}\end{center}
In this work, we explore the quantum counterpart of a recently introduced communication model: the Byzantine multiple-access channel, first studied in \cite{Sa/Ba/De/Pr} and \cite{Sa/Ba/De/Pr2}. Specifically, we investigate a classical-quantum multiple-access channel where one of the transmitters may act adversarially. Instead of transmitting legitimate messages, the adversarial sender may inject malicious inputs as a jamming attack to disrupt communication.

A key challenge in this setting is that the receiver does not know which sender is adversarial. Despite this uncertainty, we require that the messages of honest senders be reliably decoded. The capacity region of the classical Byzantine multiple-access channel was established in \cite{Sa/Ba/De/Pr}, while \cite{Sa/Ba/De/Pr2} introduced a technique for identifying the adversarial sender.

When extending communication models to the quantum domain, several new challenges arise. In quantum systems, every measurement can alter the quantum state, meaning that disruptions may not only come from an adversarial sender's inputs but also from previously applied decoding operations (cf. Section \ref{NoO}).

Moreover, while the capacity of a classical arbitrarily varying channel can be achieved using a deterministic code-without relying on any randomness (cf. \cite{Cs/Na} and \cite{Ahl/Ca}, where a deterministic code for an arbitrarily varying multiple-access channel was constructed)-the same does not hold for quantum arbitrarily varying channels, where only random coding techniques have been developed so far (cf. \cite{Ahl/Bli} and \cite{Da/Wa/Ha}).

As a result, many of the sophisticated techniques from \cite{Sa/Ba/De/Pr} cannot be directly extended to quantum Byzantine multiple-access channels. Similarly, the identification technique proposed in \cite{Sa/Ba/De/Pr2} is not applicable: while a sequence of classical symbols can be easily verified against a set of expected legal channel outputs, a "false" quantum state may still have a nonzero probability of being accepted by the decoding operators for legitimate messages.
To overcome these challenges, we must develop a novel coding approach.

\section{Preliminaries}

Let \(\rho_1\) and \(\rho_2\) be Hermitian operators on a finite-dimensional complex Hilbert space \( A \). We write \(\rho_1 \geq \rho_2\) (or equivalently, \(\rho_2 \leq \rho_1\)) if \(\rho_1 - \rho_2\) is positive semidefinite.  
For a finite-dimensional complex Hilbert space \( A \), the set of density operators on \( A \) is denoted by  
\[
\mathcal{S}(A) := \left\{ \rho \in \mathcal{L}(A) \mid \rho \text{ is Hermitian, } \rho \geq 0, \text{ and } \mathrm{tr}(\rho) = 1 \right\},
\]  
where \(\mathcal{L}(A)\) is the set of linear operators on \( A \), and \( 0_A \) denotes the null matrix on \( A \). Note that every operator in \( \mathcal{S}(A) \) is bounded.

Let \( X \) be a finite alphabet, and \( A \) be a finite-dimensional complex Hilbert space. A classical-quantum channel is a map \( \mathtt{W}: X \rightarrow \mathcal{S}(A) \), where \( x \mapsto \mathtt{W}(x) \).

For a quantum state \( \rho \in \mathcal{S}(A) \), the von Neumann entropy of \( \rho \) is denoted by \( S(\rho) = - \mathrm{tr}(\rho \log \rho) \).  

Given a set \( X \), a Hilbert space \( A \), a classical-quantum channel \( \mathtt{V}: X \rightarrow \mathcal{S}(A) \), and a probability distribution \( P \) on \( X \), the conditional von Neumann entropy is defined as  
\[
S(\mathtt{V}|P) := \sum_{x \in X} P(x) S\left( \mathtt{V}(x) \right).
\]

The quantum relative entropy between \( \rho \) and \( \sigma \) is defined as  
\[
D(\rho \parallel \sigma) := \mathrm{tr}\left( \rho \left( \log \rho - \log \sigma \right) \right),
\]  
if \( \text{supp}(\rho) \subset \text{supp}(\sigma) \), and \( D(\rho \parallel \sigma) = \infty \) otherwise.

For a set \( X \), a Hilbert space \( A \), a classical-quantum channel \( \mathtt{V}: X \rightarrow \mathcal{S}(A) \), and a probability distribution \( P \) on \( X \), the quantum mutual information of the channel \( \mathtt{V} \) with input distribution \( P \) is defined as  
\[
I(P; \mathtt{V}) := \sum_{x \in X} P(x) D\left( \mathtt{V}(x) \parallel \sum_{x' \in X} P(x') \mathtt{V}(x') \right).
\]

Let \( \Phi := \{\rho_x : x \in X\} \) be a set of quantum states labeled by elements of \( X \). For a probability distribution \( P \) on \( X \), the Holevo quantity, or Holevo information, is defined as  
\[
\chi(P; \Phi) := S\left( \sum_{x \in X} P(x) \rho_x \right) - \sum_{x \in X} P(x) S(\rho_x).
\]

For a set \( X \) and a Hilbert space \( A \), let \( \mathtt{V}: X \rightarrow \mathcal{S}(A) \) be a classical-quantum channel. For a probability distribution \( P \) on \( X \), the Holevo quantity of the channel \( \mathtt{V} \) with input distribution \( P \) is defined as  
\[
\chi(P; \mathtt{V}) := S\left( \sum_{x \in X} P(x) \mathtt{V}(x) \right) - S(\mathtt{V} | P).
\]

Notice that when we choose an "optimal" set \( \Phi := \{\rho_x : x \in X\} \) (in the sense that it maximizes the quantum mutual information and the Holevo quantity), we have (cf. \cite{Hay2})  
\[
I(P; \mathtt{V}) = \chi(P; \mathtt{V}) \quad \text{.}
\]

\begin{definition}\label{symmet}
A set of classical-quantum channels \( \{\mathbf{W}(\cdot, t) : t \in \Theta \} \) is called an arbitrarily varying classical-quantum channel when the state \( t \) varies from symbol to symbol in an arbitrary manner.

We say that the arbitrarily varying classical-quantum channel \( \{\mathbf{W}(\cdot, t) : t \in \Theta \} \) is symmetrizable if there exists a parametric set of distributions \( \{\tau(\cdot \mid x) : x \in X \} \) on \( \Theta \) such that for all \( x, x' \in X \), the following condition holds:  
\[
\sum_{t \in \Theta} \tau(t \mid x) \mathbf{W}(x', t) = \sum_{t \in \Theta} \tau(t \mid x') \mathbf{W}(x, t).
\]
For a probability distribution $p$ on $\Theta$, we denote $\mathbf{W}(\cdot, p):=\sum_{t}\mathbf{W}(\cdot, t)$.
\end{definition}

\begin{definition}
Consider \( k \) independent senders, with sender \( i \) using an alphabet \( X_i \) and an a priori probability distribution \( P_i \). Let \( A \) be a finite-dimensional complex Hilbert space.  

A \( k \)-user Byzantine multiple-access classical-quantum channel \( \mathtt{W} \) is a map  
\[
\mathtt{W}: X_1 \times \cdots \times X_k \rightarrow \mathcal{S}(A),
\]  
where one of the senders may be adversarial, but the decoder does not know a priori which sender is adversarial.

\end{definition}

 \begin{definition}
A random code for a \( k \)-user Byzantine multiple-access classical-quantum channel \( \mathtt{W} \) consists of an encoder  
\[
\Bigl\{ \{ x_i^n(m_i)^{\gamma_i} : \gamma_i \in \Lambda_i \} : m_i \in M_i \Bigr\} \subset X_i^n
\]  
for every sender \( i \in \{1, \dots, k\} \), and a set of positive operator-valued measure (POVM) decoder operators  
\[
\Bigl\{ \left\{ D_{m_i}^{\gamma_i} : m_i \in M_i \right\} : \gamma_i \in \Lambda_i \Bigr\}
\]  
for each sender's message. These are parameterized by a random variable on a finite set \( \Lambda_i \). The interpretation of this definition is that both sender \( i \) and the receiver have access to the outcome of a random experiment uniformly distributed on \( \Lambda_i \).

The average probability of decoding error for the message of sender \( i \) is defined as  
\[
1 - \frac{1}{J_n} \frac{1}{|\Lambda_i|} \sum_{m_i \in M_i} \sum_{\gamma_i \in \Lambda_i} \mathrm{tr} \left( {\mathtt{W}}^{\otimes n} \left( x_i^n(m_i)^{\gamma_i} \right) D_{m_i}^{\gamma_i} \right).
\]
\label{randef}
\end{definition}

\begin{definition}
A non-negative number \( R_i \) is an achievable rate for the message of sender \( i \) in a Byzantine multiple-access classical-quantum channel \( \mathtt{W} \) under random coding if, for every \( \delta > 0 \), \( \epsilon > 0 \), and sufficiently large \( n \), there exists a random code such that  
\[
\frac{\log |M_i|}{n} > R_i - \delta,
\]  
and the average probability of decoding error for the message of sender \( i \) is less than \( \epsilon \) for any input sequence from the adversarial sender.

The closure of all achievable random rates for \( \mathtt{W} \) is called the random capacity region of \( \mathtt{W} \).
\end{definition}

\section{Message Transmission over
Quantum Byzantine Multiple Access Channel} 

\subsection{Impact of Decoding Order}\label{NoO}
In the classical multiple-access channel scenario, it is evident that the order in which the messages of the senders are decoded does not matter. In fact, most classical multiple-access codes decode all messages simultaneously (cf. \cite{Ahl4} and \cite{Lia}). However, in the quantum multiple-access channel scenario, the previously decoded operator can affect the subsequent decoded messages, as every measurement may alter the channel outputs.

The standard solution to this issue in the quantum multiple-access channel scenario is based on the following observation. Suppose the first decoding operator can decode the first message with a probability sufficiently close to 1, and assume this decoding operator is constructed using subspace projections. In this case, we can apply the Tender Operator Lemma (cf. \cite{Win} and \cite{Og/Na}) to show that this decoding operator does not significantly alter the channel outputs.

However, in the Byzantine multiple-access classical-quantum channel scenario, we cannot apply the Tender Operator trick, since we cannot guarantee that the first decoding operator will decode the first message with a sufficiently high probability, as the first sender may indeed be the adversarial one. Therefore, in this scenario, it is possible that if the first sender is adversarial, the first decoding operator may alter the channel outputs in such a way that the second message becomes undecodable. The following simple one-shot
example illustrates this point:

We define the input alphabets as \( X_1 := \{0,1\} \) and \( X_2 := \{0,1,2\} \). Let \( A \otimes B \) be a six-dimensional Hilbert space spanned by the orthonormal basis:

\[
\{|i\rangle^A\otimes|j\rangle^B : i \in \{0,1\}, j \in \{0,1,2\} \}.
\]

Consider a two-user Byzantine multiple access classical-quantum channel \( \mathtt{W}: X_1 \times X_2 \to \mathcal{S}(A \otimes B) \), defined by

\[
\mathtt{W} ((i,j)) = (|i\rangle^A\otimes |j \rangle^B )
(\langle j|^B \otimes \langle i |^A), \quad \forall i \in \{0,1\}, j \in \{0,1,2\}.
\]

Suppose sender 1 transmits messages from the set \( \{m_0^1, m_1^1\} \), and sender 2 transmits messages from the set \( \{m_0^2, m_1^2\} \). That is, each sender's message set consists of two possible messages.

\subsection*{Encoding Strategy}

Each sender encodes their message as follows:  
For \( k \in \{1,2\} \), sender \( k \) simply transmits \( i \in \{0,1\} \) as their channel input to represent message \( m_i^k \).

\subsection*{Decoding Strategy}

The receiver employs the following POVM to decode sender 1's message:

\[
D_0^{(1)} := \sum_{j=0}^{2} (|0\rangle^A \otimes |j\rangle^B)(\langle j|^B \otimes \langle 0|^A),
\]

\[
D_1^{(1)} := \sum_{j=0}^{2} (|1\rangle^A \otimes |j\rangle^B)(\langle j|^B \otimes \langle 1|^A).
\]

In other words, the receiver measures only the \( A \)-part of the channel output, leaving the \( B \)-part unchanged.

The receiver employs the following POVM to decode the message of sender 2:

\begin{align*}
D_0^{(2)} &= (|0\rangle^A\otimes |0 \rangle ^B)(\langle 0|^B \otimes \langle 0 |^A)
+ (|1\rangle^A\otimes |0 \rangle^B)(\langle 0| ^B\otimes \langle 1 |^A) \\
&\quad + \frac{1}{2}(|0\rangle^A\otimes |2\rangle^B)(\langle 2|^B \otimes \langle 0 |^A)
+ \frac{1}{2}(|1\rangle^A\otimes |2 \rangle^B)(\langle 2|^B \otimes \langle 1 |^A) \\
&\quad + \frac{1}{2}(|0\rangle^A\otimes |2 \rangle^B)(\langle 2|^B \otimes \langle 1 |^A)
+ \frac{1}{2}(|1\rangle^A\otimes |2\rangle^B)(\langle 2|^B \otimes \langle 0 |^A), \\[0.3cm]
D_1^{(2)} &= (|0\rangle^A\otimes |1 \rangle ^B)(\langle 1|^B \otimes \langle 0 |^A)
+ (|1\rangle^A\otimes |1 \rangle^B)(\langle 1| ^B\otimes \langle 1 |^A) \\
&\quad + \frac{1}{2}(|0\rangle^A\otimes |2 \rangle^B)(\langle 2|^B \otimes \langle 0 |^A)
+ \frac{1}{2}(|1\rangle^A\otimes |2\rangle^B)(\langle 2|^B \otimes \langle 1 |^A) \\
&\quad - \frac{1}{2}(|0\rangle^A\otimes |2 \rangle^B)(\langle 2|^B \otimes \langle 1 |^A)
- \frac{1}{2}(|1\rangle^A\otimes |2\rangle^B)(\langle 2|^B \otimes \langle 0 |^A).
\end{align*}

The receiver applies these operators to decode the message of sender 2.

\begin{center}\begin{figure}[H]
\begin{minipage}{0.45\linewidth}
        \centering
\includegraphics[width=0.95\linewidth]{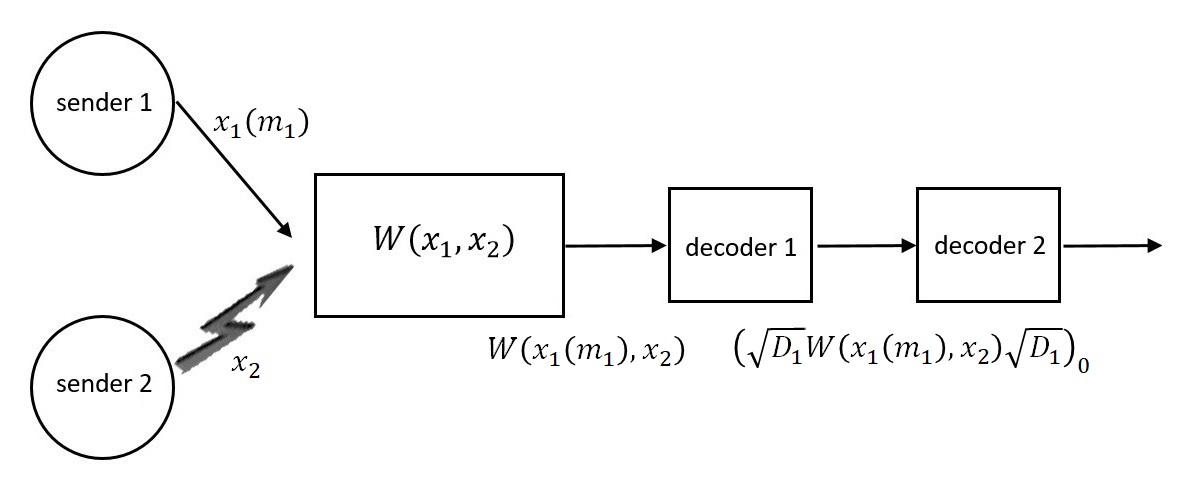}
\caption{Sender \(1\) is trustworthy, while sender \(2\) is adversarial. The receiver first decodes the message from sender \(1\), followed by the message from sender \(2\).}
 \end{minipage}
  \hspace{.05\linewidth}
\begin{minipage}{0.45\linewidth}
        \centering
\includegraphics[width=0.95\linewidth]{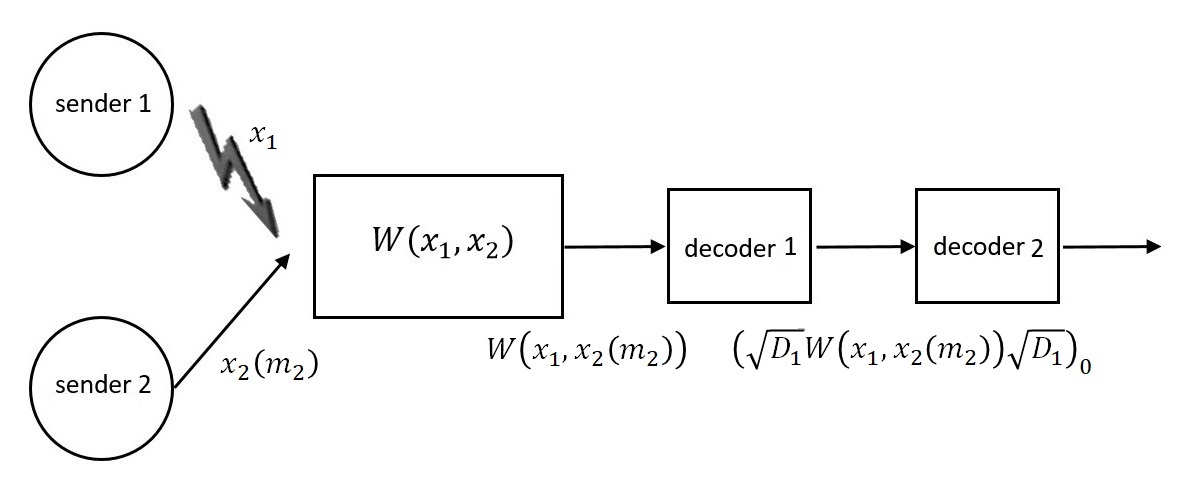}
\caption{Sender \(2\) is trustworthy, while sender \(1\) is adversarial. The receiver first decodes the message from sender \(1\), followed by the message from sender \(2\).
}
 \end{minipage}
\label{2userbymac}
\end{figure}\end{center}

First, assume that the receiver decides to decode the message of sender \(1\) first.
We distinguish three different cases:

\textbf{Case 1a:} Sender \(1\) is trustworthy, while sender \(2\) is adversarial.  
In this case, the receiver can decode sender \(1\)'s message with probability \(1\), regardless of the input from sender \(2\).

\textbf{Case 1b:} Both senders are trustworthy.  
The receiver first decodes sender \(1\)'s message with probability \(1\). Since this decoding does not alter part \(B\) of the channel outputs, the receiver can then decode sender \(2\)'s message with probability \(1\) as well.

\textbf{Case 1c:} Sender \(1\) is adversarial, while sender \(2\) is trustworthy.  
Since the measurement operators \(D_0^{(1)}\) and \(D_1^{(1)}\) do not affect part \(B\) of the channel outputs, the receiver can still decode sender \(2\)'s message with probability \(1\), regardless of the input from sender \(1\).  

Thus, in this decoding order, we obtain a reliable code for the Byzantine multiple-access channel.

\begin{center}\begin{figure}[H]
\begin{minipage}{0.45\linewidth}
        \centering
\includegraphics[width=0.95\linewidth]{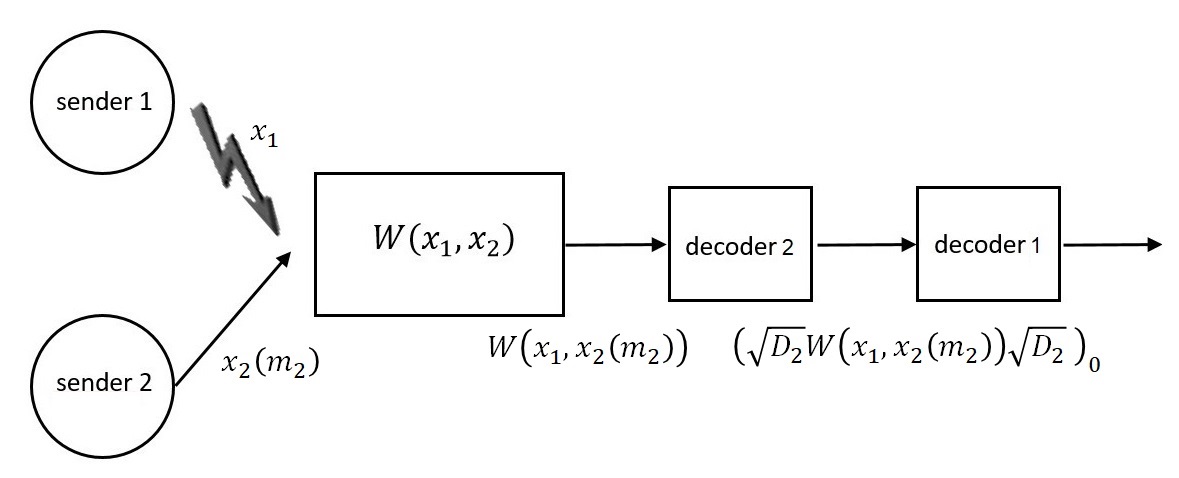}
\caption{\textit{Sender \(2\) is trustworthy, while sender \(1\) is adversarial. The receiver first decodes sender \(2\)'s message, followed by sender \(1\)'s message.}
}
 \end{minipage}
  \hspace{.05\linewidth}
\begin{minipage}{0.45\linewidth}
        \centering
\includegraphics[width=0.95\linewidth]{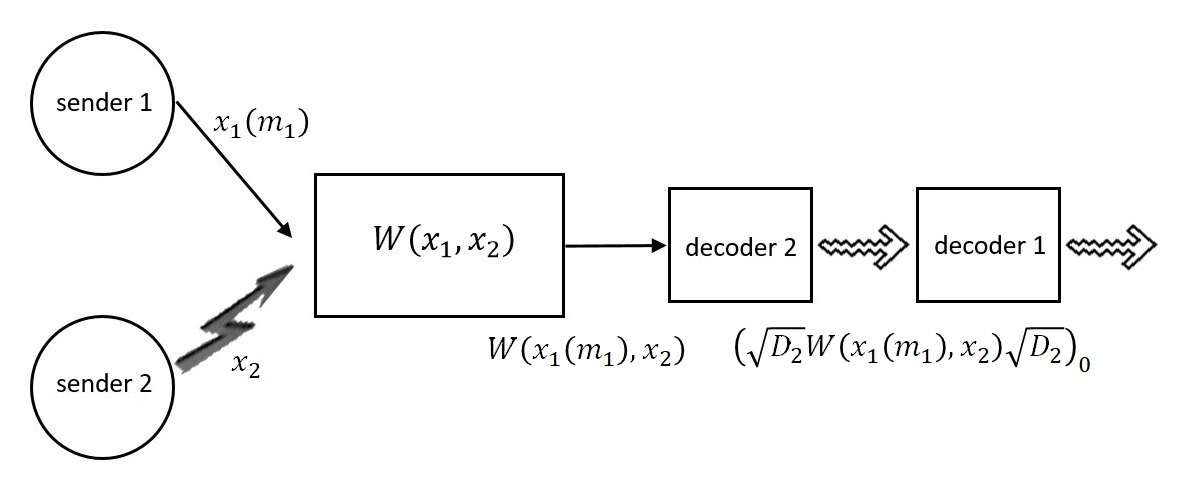}
\caption{\textit{Sender \(1\) is trustworthy, while sender \(2\) is adversarial. The receiver first decodes sender \(2\)'s message, followed by sender \(1\)'s message. In this case, the decoding operator for sender \(2\)'s message interferes with the decoding of sender \(1\)'s message.}
}
 \end{minipage}
\label{2userbymac2}
\end{figure}\end{center}

Now, let us assume that the receiver decides to decode the message of sender $2$ first. Once again, we distinguish three different cases:

\textbf{Case 2a}: Sender $2$ is trustworthy, but sender $1$ is adversarial.  
In this case, the receiver can decode sender $2$'s message with probability 1, regardless of what sender $1$ inputs into the channel.

\textbf{Case 2b}: Both senders are trustworthy.  
Initially, the receiver can decode sender $2$'s message with probability 1. Since $D_0^{(2)}$ and $D_1^{(2)}$ do not affect part $A$ of the channel outputs when sender $2$ is trustworthy, the receiver can then decode sender $1$'s message with probability 1.

\textbf{Case 2c}: Sender $2$ is adversarial, while sender $1$ is trustworthy.  
Consider the situation where sender $1$ wants to send $m_{0}^1$ and transmits $0$ through the channel, while sender $2$, as part of a jamming attack, transmits $2$. The channel input is now $(0,2)$, and the channel output is:

\[
(|0\rangle^A \otimes |2 \rangle^B)(\langle 2|^B \otimes \langle 0 |^A)
\]

When the receiver, unaware that sender $2$ is adversarial, uses $D_0^{(2)}$ and $D_1^{(2)}$ to measure sender $2$'s message, there is a $50\%$ chance that the receiver obtains
 either $m_{0}^2$ or $m_{1}^2$. When he obtains
 $m_{0}^2$, then the resulting channel output after the measurement is given by:

\[
\frac{\sqrt{D_0^{(2)}}(|0\rangle^A \otimes |2 \rangle^B)(\langle 2|^B \otimes \langle 0 |^A) \sqrt{D_0^{(2)}}}
{\mathsf{Tr}(D_0^{(2)} (|0\rangle^A \otimes |2 \rangle^B))}
= \frac{1}{2}(|0\rangle^A \otimes |2\rangle^B)(\langle 2|^B \otimes \langle 0 |^A)
+ \frac{1}{2}(|1\rangle^A \otimes |2 \rangle^B)(\langle 2|^B \otimes \langle 1 |^A)
\]

If instead the receiver obtains $m_{1}^2$, the channel output after the measurement is:

\[
\frac{\sqrt{D_1^{(2)}}(|0\rangle^A \otimes |2 \rangle^B)(\langle 2|^B \otimes \langle 0 |^A) \sqrt{D_1^{(2)}}}
{\mathsf{Tr}(D_1^{(2)} (|0\rangle^A \otimes |2 \rangle^B))}
= \frac{1}{2}(|0\rangle^A \otimes |2\rangle^B)(\langle 2|^B \otimes \langle 0 |^A)
- \frac{1}{2}(|0\rangle^A \otimes |2 \rangle^B)(\langle 2|^B \otimes \langle 1 |^A)
\]

When the receiver subsequently uses $D_0^{(1)}$ and $D_1^{(1)}$ to measure sender $1$'s message, there is a $50\%$ chance that the receiver obtains either $m_{0}^1$ or $m_{1}^1$.

The same scenario occurs when sender $1$ transmits $m_{1}^1$ (i.e., when $1$ is sent through the channel) and sender $2$ jams the channel by sending $2$. In this case, the "worse order", where the adversarial message is decoded first, results in the second code becoming ineffective due to the interference of the first POVM decoding operator. Therefore, in case $2c$, decoding the message of the adversarial sender first renders the code unreliable.

\subsection{Main Result}

For our channel model, we should employ a separable decoding operator for each sender's message, rather than a simultaneous decoding approach, since not all senders can be trusted. In fact, the standard methods of quantum simultaneous decoding also utilize the same technique as the Tender Operator Lemma (cf. \cite{Fa/Ha/Sa/Se/Wi}, \cite{Se}). This implies that these ``simultaneous decoding operators'' essentially decode one message at a time, with a few exceptions (cf. \cite{Ch/Ne/Se}).

The standard coding strategy for non-symmetrizable classical-quantum arbitrarily varying channels involves constructing a random codeword (cf. \cite{Bl/Br/Th2}, see also \cite{Mo}, \cite{Hay}). This concept was further utilized in \cite{Da/Wa/Ha} to develop a code for classical-quantum arbitrarily varying multiple-access channels.

Now, assume the receiver decodes the messages in the following order: the message from sender 1 is decoded first, followed by the message from sender 2. As shown in \cite{Ahl/Bli}, by setting \( |M_1| < \max_{x^n_2} \chi(p_1, B_{x^n_2}) - \delta \), a random code \( E_1 = \{ x^n_1(m_1)^{\gamma_1}, D_{m_1}^{\gamma_1} : m_1 \in M_1 \} \) exists such that

\[
\max_{x^n_2} \frac{1}{|M_1|} \sum_{m_1} \mathtt{W}((x^n_1(m_1), x^n_2)) D_{m_1}^{\gamma_1} > 1 - \epsilon,
\]

where, for each \( m_1 \), the elements \( \{ x^n_1(m_1)^{\gamma_1} : \gamma_1 \} \) are permutations of a specific codeword \( x^n_1(m_1) \), and \( \{ D_{m_1}^{\gamma_1} : \gamma_1 \} \) are products of permutation operators and a single decoding operator \( D_{m_1} \). 

Thus, when the receiver receives the output corresponding to the common randomness shared with sender 1, he first applies a permutation, performs the initial measurement, and then reverses the permutation before making the second measurement. Consequently, the common randomness shared with sender 1 does not affect sender 2's coding strategy. Sender 2 only needs to consider the code \( \{ x^n_1(m_1), D_{m_1} : m_1 \in M_1 \} \), rather than every permutation of the codeword \( \{ x^{\gamma_1}_1(m_1), D_{m_1}^{\gamma_1} : m_1 \in M_1 \} \).

The POVM \( \{ D_{m_1} : m_1 \in M_1 \} \) defines a channel \( C_{E_1} \), where

\[
\rho \rightarrow \sum_{m_1} \sqrt{D_{m_1}} \rho \sqrt{D_{m_1}}.
\]

\begin{theorem}
Let \( \mathtt{W} \) denote a classical-quantum 2-user Byzantine multiple-access channel. Assume that the message of sender 1 is decoded first, followed by the message of sender 2. The random capacity region of \( \mathtt{W} \) is given by

\[
R_1 \leq \max_{p_1} \min_{p_2} I(p_1, B),
\]

\[
R_2 \leq \max_{p_2} \min_{p_1} I(p_2, E_1 B),
\]

where 
\(B\)  represents the outcome of \(  \mathtt{W}((p_1, p_2) \). and 
\( E_1 B \) represents the outcome of \( C_{E_1} \circ \mathtt{W}(p_1,p_2) \).

\label{lwbactheo} \end{theorem}

\begin{proof}
The direct part of the proof follows directly from the arbitrarily varying capacity of \( \mathtt{W}((x^n_1(i), \cdot)) \) after the measurement of sender 1's message (cf. \cite{Ahl/Bli} and \cite{Da/Wa/Ha}). It is also evident that for our coding scheme, the rate is upper-bounded by (\ref{lwbactheo}). 

Now, let us consider the converse in the general case, i.e., to show that regardless of the coding strategy we choose, the order of decoding always affects the capacity. Specifically, for the message decoded second, we must always consider \( C_{E_1} \circ \mathtt{W} \) for some \( C_{E_1} \).

If the receiver uses a simultaneous decoding operator \( \{D_{m_1, m_2}^{\gamma_1, \gamma_2} : m_1, m_2\} \), we can always construct separable decoding operators \( \{D_{m_1}^{\gamma_1, \gamma_2} : m_1\} \) and \( \{D_{m_2}^{\gamma_1, \gamma_2} : m_2\} \) by defining
\[
D_{m_1}^{\gamma_1, \gamma_2} := \sum_{m_2} D_{m_1, m_2}^{\gamma_1, \gamma_2}
\quad \text{and} \quad
D_{m_2}^{\gamma_1, \gamma_2} := \sum_{m_1} D_{m_1, m_2}^{\gamma_1, \gamma_2}.
\]
Thus, without loss of generality, we may assume that the receiver uses a separable decoding strategy, decoding sender 1's message first with a POVM \( \{D_{m_1}^{\gamma_1, \gamma_2} : m_1\} \). After the measurement of sender 1's message, the channel output will result in 
\[
\mathrm{tr}\left( \mathtt{W}((x^n_{1}, x^n_{2})) D_{m_1}^{\gamma_1, \gamma_2} \right)^{-1}
\sqrt{D_{m_1}^{\gamma_1, \gamma_2}} \mathtt{W}((x^n_{1}, x^n_{2})) \sqrt{D_{m_1}^{\gamma_1, \gamma_2}},
\]
with probability \( \mathrm{tr}(D_{m_1}^{\gamma_1, \gamma_2} \mathtt{W}((x^n_{1}, x^n_{2}))) \).

Thus, after the measurement of sender 1's message, when the receiver begins to decode sender 2's message, the message transmission rate cannot exceed the arbitrarily varying channel capacity of 
\( \mathtt{W}(p_1,p_2) \) within the subspace defined by \( D_{m_1}^{\gamma_1, \gamma_2} \). The arbitrarily varying channel capacity of \( \mathtt{W}((p_1, \cdot)) \) into this subspace is equal to
\[
\min_{x^n_1} \chi(p_2, \mathrm{tr}(D_{m_1}^{\gamma_1, \gamma_2} \mathtt{W}((p_1, x^n_2)))^{-1} \sqrt{D_{m_1}^{\gamma_1, \gamma_2}} B_{p_1, p_2} \sqrt{D_{m_1}^{\gamma_1, \gamma_2}}).
\]

Now, we define \( C_{E_1} \) by
\[
C_{E_1}(\rho) = \sum_{m_1} \sqrt{D_{m_1}^{\gamma_1, \gamma_2}} \rho \sqrt{D_{m_1}^{\gamma_1, \gamma_2}}.
\]
In total, the expected capacity cannot exceed
\[
\sum_{m_1} \mathrm{tr}\left( D_{m_1}^{\gamma_1, \gamma_2} \mathtt{W}((x^n_1, x^n_2)) \right)
\min_{x^n_1} \chi\left( p_2, \mathrm{tr}\left( D_{m_1}^{\gamma_1, \gamma_2} \mathtt{W}((x^n_1, x^n_2)) \right)^{-1} \sqrt{D_{m_1}^{\gamma_1, \gamma_2}} B \sqrt{D_{m_1}^{\gamma_1, \gamma_2}} \right),
\]
which simplifies to
\[
\min_{p_1} I(p_2, E_1 B).
\]

\end{proof}\vspace{0.2cm}

Let us now consider a Byzantine multiple-access channel with 3 senders, where at most one sender is adversarial. By Theorem \ref{lwbactheo}, we can construct reliable random codes \( \{ x^{\gamma_1}_1(m_1), D_{m_1}^{\gamma_1} : m_1 \in M_1 \} \), which are permutations of a code \( \{ x^n_1(m_1), D_{m_1} : m_1 \in M_1 \} \), and \( \{ x^{\gamma_2}_2(m_2), D_{m_2}^{\gamma_2} : m_2 \in M_2 \} \), which are permutations of a code \( \{ x^n_2(m_2), D_{m_2} : m_2 \in M_2 \} \). 

We define \( C_{E_1} \) and \( C_{E_2} \) in the same way as in Theorem \ref{lwbactheo}, i.e.,
\[
C_{E_2}(\rho) := \sum_{m_2} \mathrm{tr}(D_{m_2} \rho)^{-1} \sqrt{D_{m_2}} \rho \sqrt{D_{m_2}}.
\]

\begin{theorem}
Let \( \mathtt{W} \) be a classical-quantum 3-user Byzantine multiple-access channel. Assume that the receiver decodes the messages in the following order: the message of sender 1, then the message of sender 2, and finally the message of sender 3.

The capacity region of \( \mathtt{W} \) is  the following:

\begin{align}
R_1 &\leq \min \left( \max_{p_1} \min_{p_2} I(p_1, B), \max_{p_1} \min_{p_3} I(p_1, B) \right) \\
R_2 &\leq \min \left( \max_{p_2} \min_{p_1} I(p_2, E_1 B), \max_{p_2} \min_{p_3} I(p_2, B) \right) \\
R_3 &\leq \min \left\{ \max_{p_3} \min_{p_1} I(p_3, E_1 B | p_2), \max_{p_3} \min_{p_2} I(p_3, E_2 B | p_1) \right\}
\end{align}

\label{lwbactheo3} \end{theorem}

\begin{proof}
The capacity region for the messages of sender 1 and sender 2 follow from Theorem \ref{lwbactheo}. Now, consider the situation where the receiver has already performed measurements using the codes \( \{x^{\gamma_1}_1(m_1), D_{m_1}^{\gamma_1}: m_1 \in M_1\} \), which are permutations of the code \( \{x^n_1(m_1), D_{m_1}: m_1 \in M_1\} \), and \( \{x^{\gamma_2}_2(m_2), D_{m_2}^{\gamma_2}: m_2 \in M_2\} \), which are permutations of the code \( \{x^n_2(m_2), D_{m_2}: m_2 \in M_2\} \). The receiver now aims to decode the message of sender 3.

Let \( \Theta := \{ x^n_1 \in X^n_1, x^n_2 \in X^n_2 \} \) be the set of uncertainty. We define an arbitrarily varying classical-quantum channel \( \{ \mathbf{W}(\cdot, t) : t \in \Theta \} \) as follows:

\begin{align*}
\mathbf{W}(x^n_3, x^n_2) &= \sum_{m_1} \sum_{m_2} p(x^n_1(m_1) | x^n_3) \sqrt{D_{m_2}} \, \mathtt{W}(x^n_1(m_1), x^n_2, x^n_3) \, \sqrt{D_{m_2}} \\
\text{and} \\
\mathbf{W}(x^n_3, x^n_1) &= \sum_{m_2} \sum_{m_1} p(x^n_2(m_2) | x^n_3) \sqrt{D_{m_1}} \, \mathtt{W}(x^n_1, x^n_2(m_2), x^n_3) \, \sqrt{D_{m_1}} \text{ .}
\end{align*}

Since the receiver has already measured the messages of sender 1 and sender 2, we can write
\[
\sum_{m_1} \sum_{m_2} \sqrt{D_{m_1}} \sqrt{D_{m_2}} \, \mathtt{W}(\cdot, \cdot, x^n_3) \, \sqrt{D_{m_1}} \sqrt{D_{m_2}}
\]
as
\[
\bigoplus_{m_1} \left( \sum_{m_2} \sqrt{D_{m_2}} \, \sqrt{D_{m_1}} \, \mathtt{W}(\cdot, \cdot, x^n_3) \, \sqrt{D_{m_1}} \, \sqrt{D_{m_2}} \right),
\]
and as
\[
\bigoplus_{m_2} \left( \sum_{m_1} \sqrt{D_{m_2}} \, \sqrt{D_{m_1}} \, \mathtt{W}(\cdot, \cdot, x^n_3) \, \sqrt{D_{m_1}} \, \sqrt{D_{m_2}} \right).
\]
Thus, we have
\begin{align}
\chi(p_3, \mathbf{W}(p_3, x^n_1)) &= S \left( \sum_{x^n_3} \sum_{m_2} p(x^n_3) p(x^n_2(m_2) | x^n_3) \sqrt{D_{m_2}} \, \sqrt{D_{m_1}} \, \mathtt{W}(x^n_1, x^n_2(m_2), x^n_3) \, \sqrt{D_{m_1}} \, \sqrt{D_{m_2}} \right) \notag \\
&\quad - \sum_{x^n_3} p(x^n_3) S \left( \sum_{m_2} p(x^n_2(m_2) | x^n_3) \sqrt{D_{m_2}} \, \sqrt{D_{m_1}} \, \mathtt{W}(x^n_1, x^n_2(m_2), x^n_3) \, \sqrt{D_{m_1}} \, \sqrt{D_{m_2}} \right) \notag \\
&= S \left( \sum_{x^n_3} \bigoplus_{m_2} p(x^n_2(m_2), x^n_3) \sqrt{D_{m_2}} \, \sqrt{D_{m_1}} \, \mathtt{W}(x^n_1, x^n_2(m_2), x^n_3) \, \sqrt{D_{m_1}} \, \sqrt{D_{m_1}} \right) \notag \\
&\quad - \sum_{x^n_3} p(x^n_3) S \left( \bigoplus_{m_2} p(x^n_2(m_2) | x^n_3) \sqrt{D_{m_2}} \, \sqrt{D_{m_1}} \, \mathtt{W}(x^n_1, x^n_2(m_2), x^n_3) \, \sqrt{D_{m_1}} \, \sqrt{D_{m_2}} \right) \notag \\
&= \sum_{m_2} p(x^n_2(m_2)) S \left( \sum_{x^n_3} p(x^n_3) p(x^n_3 | x^n_2(m_2)) \sqrt{D_{m_1}} \, \sqrt{D_{m_1}} \, \mathtt{W}(x^n_1, x^n_2(m_2), x^n_3) \, \sqrt{D_{m_1}} \, \sqrt{D_{m_2}} \right) \notag \\
&\quad - \sum_{x^n_3} \sum_{m_2} p(x^n_3) p(x^n_2(m_2) | x^n_3) S \left( \sqrt{D_{m_2}} \, \sqrt{D_{m_1}} \, \mathtt{W}(x^n_1, x^n_2(m_2), x^n_3) \, \sqrt{D_{m_1}} \, \sqrt{D_{m_2}} \right) \notag \\
&= \sum_{m_2} p(x^n_2(m_2)) S \left( \sum_{x^n_3} p(x^n_3) p(x^n_3 | x^n_2(m_2)) \sqrt{D_{m_1}} \, \sqrt{D_{m_1}} \, \mathtt{W}(x^n_1, x^n_2(m_2), x^n_3) \, \sqrt{D_{m_1}} \, \sqrt{D_{m_2}} \right) \notag \\
&\quad - \sum_{m_2} p(x^n_2(m_2)) \sum_{x^n_3} p(x^n_3 | x^n_2(m_2)) S \left( \sqrt{D_{m_2}} \, \sqrt{D_{m_1}} \, \mathtt{W}(x^n_1, x^n_2(m_2), x^n_3) \, \sqrt{D_{m_1}} \, \sqrt{D_{m_2}} \right) \notag \\
&= I(p_3, E_1 E_2 B | p_2) \label{cp3w3p3x1}.
\end{align}
Similarly, by analogous arguments, we obtain
\[
\chi(p_3, \mathbf{W}(p_3, x^n_2)) = I(p_3, E_1 E_2 B | p_1) \label{cp3w3p3x2}.
\]

Denote the set of probability distributions on $\Theta$ by $P(\Theta)$.
Notice that when the code for the legitimate sender is reliable, then, by the Tender Operator Lemma, only the jammer's decoding operator affects the channel output. Therefore, by \cite{Ahl/Bli} and \cite{Da/Wa/Ha}, and if we set
\begin{align}
|M_3| &< \max_{p_3} \min_{p_t \in P(\Theta)} \chi(p_3, \mathbf{W}(p_3, p_t)) \notag \\
&= \min \left( \max_{p_3} \min_{p_1} \chi(p_3, \mathbf{W}(p_3, p_1)), \max_{p_3} \min_{p_2} \chi(p_3, \mathbf{W}(p_3, p_2)) \right) \notag \\
&= \min \left\{ \max_{p_3} \min_{p_1} I(p_3, E_1 B | p_2), \max_{p_3} \min_{p_2} I(p_3, E_2 B | p_1) \right\}, \notag
\end{align}
for sufficiently large $n$, there exists a reliable random code
\[
\left\{ \left\{ x^{\gamma_3}_3(m_3), D_{m_3}^{\gamma_3} : m_3 \in M_3 \right\} : \gamma_3 \right\}
\]
a permutation of a code $\{ x^n_3(m_3), D_{m_3} : m_3 \in M_3 \}$, such that for any positive $\epsilon$ and sufficiently large $n$,
\begin{align}
1 - \epsilon &< \max_{x^n_1 \in X^n_1} \frac{1}{|\Lambda_3|} \sum_{\gamma_3 \in \Lambda_3} \mathrm{tr} \left( \sum_{m_3} p(m_3) \mathbf{W}(x^{\gamma_3}_3(m_3), x^n_1) D_{m_3}^{\gamma_3} \right) \notag \\
&= \max_{x^n_1 \in X^n_1} \frac{1}{|\Lambda_3|} \sum_{\gamma_3 \in \Lambda_3} \mathrm{tr} \left( \sum_{m_3} \sum_{m_2} p(m_3) p(x^n_2(m_2) | x^{\gamma_3}_3(m_3)) C_{E_1} \mathtt{W}(x^n_1, x^n_2(m_2), x^{\gamma_3}_3(m_3)) D_{m_3}^{\gamma_3} \right), \notag
\end{align}
and
\begin{align}
1 - \epsilon &< \max_{x^n_2 \in X^n_2} \frac{1}{|\Lambda_3|} \sum_{\gamma_3 \in \Lambda_3} \mathrm{tr} \left( \sum_{m_3} p(m_3) \mathbf{W}(x^{\gamma_3}_3(m_3), x^n_2) D_{m_3}^{\gamma_3} \right) \notag \\
&= \max_{x^n_2 \in X^n_2} \frac{1}{|\Lambda_3|} \sum_{\gamma_3 \in \Lambda_3} \mathrm{tr} \left( \sum_{m_3} \sum_{m_1} p(m_3) p(x^n_1(m_1) | x^{\gamma_3}_3(m_3)) C_{E_2} \mathtt{W}(x^n_1(m_1), x^n_2, x^{\gamma_3}_3(m_3)) D_{m_3}^{\gamma_3} \right). \notag
\end{align}

Let us first consider the case where sender 1 is the adversarial one and sends $x^n_1$ as a jamming attack. The legitimate senders, sender 2 and sender 3, send $x^{\gamma_2}_2$ and $x^{\gamma_3}_3$, respectively. The receiver initially performs a measurement with $\{ D_{m_1}^{\gamma_1} : m_1 \}$, and then reverses the permutation $\gamma_1$. The resulting channel output occurs with probability
\[
\mathrm{tr} \left( D_{m_1} \mathtt{W} \left( (x^n_1, x^{\gamma_2}_2(m_2), x^{\gamma_3}_3(m_3)) \right) \right),
\]
and the output is given by
\[
\mathrm{tr} \left( D_{m_1} \mathtt{W} \left( (x^n_1, x^{\gamma_2}_2(m_2), x^{\gamma_3}_3(m_3)) \right) \right)^{-1} \sqrt{D_{m_1}} \mathtt{W} \left( (x^n_1, x^{\gamma_2}_2(m_2), x^{\gamma_3}_3(m_3)) \right) \sqrt{D_{m_1}}.
\]

Next, the receiver performs a measurement with $\{ D_{m_2}^{\gamma_2} : m_2 \}$ and then reverses the permutation $\gamma_2$. By the Tender Operator Lemma, the channel output changes negligibly. Afterward, he performs a measurement with $\{ D_{m_3}^{\gamma_3} : m_3 \}$ and obtains $M_3$ with probability
\[
\frac{1}{|\Lambda_3|} \sum_{\gamma_3} \mathrm{tr} \left( D_{m_1} \mathtt{W} \left( (x^n_1, x^{\gamma_2}_2(m_2), x^{\gamma_3}_3(m_3)) \right) \right)^{-1} 
\mathrm{tr} \left( C_{E_1} \circ \mathtt{W} \left( (x^n_1, x^{\gamma_2}_2(m_2), x^{\gamma_3}_3(m_3)) \right) D_{m_3}^{\gamma_3} \right)^{-1}.
\]
Thus, the expected value of correctly decoding $m_3$ is
\[
\frac{1}{|\Lambda_3|} \sum_{\gamma_3 \in \Lambda_3} \mathrm{tr} \left( \sum_{m_3} \sum_{m_1} p(m_3) p(x^n_1(m_1) | x^{\gamma_3}_3(m_3)) C_{E_2} 
\mathtt{W} \left( (x^n_1(m_1), x^n_2, x^{\gamma_3}_3(m_3)) \right) D_{m_3}^{\gamma_3} \right).
\]

Therefore, this code is a reliable random code for sender 3 to transmit his message through $\mathtt{W}$.\vspace{0.2cm}

For the converse for the message of sender 3, 
we may, without loss of generality,  assume that the receiver uses a separable decoding strategy,
as stated in the proof of Thoerem 
\ref{lwbactheo} .
We consider now the case when sender 3 sends
$x^{\gamma_3}_3(m_3)$
through the channel.
When sender 1 is 
adversarial and 
sends $x^n_1$ as jamming attack,
bthen by the Tender Operator Lemma,
the channel output will be
sufficiently closed to
\[ \sum_{m_2} p(x^n_1(m_1) | x^{\gamma_3}_3(m_3)) \sqrt{D_{m_1}} \, \mathtt{W}(x^n_1(m_1), x^n_2, x^{\gamma_3}_3(m_3)) \, \sqrt{D_{m_1}} \]
with probability
$\mathrm{tr} \biggl(
 \sum_{m_2} p(x^n_1(m_1) | x^{\gamma_3}_3(m_3))$ $ \sqrt{D_{m_1}}$ $\mathtt{W}(x^n_1(m_1), x^n_2, x^{\gamma_3}_3(m_3)) $ $\sqrt{D_{m_1}} \biggr)$.
Thus, the expectetd value of 
the channel output will be
 closed to the arbitrarily varying classical-quantum channel
output 
$\mathbf{W}(x^{\gamma_3}_3(m_3), x^n_1))$,
When sender 2 is 
adversarial and 
sends $x^n_2$ as jamming attack,
the expectetd value of 
the channel output will be
 closed to the arbitrarily varying classical-quantum channel
output 
$\mathbf{W}(x^{\gamma_3}_3(m_3), x^n_2))$.
By \cite{Ahl/Bli} and \cite{Da/Wa/Ha}, 
the rate for the message of sender 3
cannot exceed
\[ \max_{p_3} \min_{p_1} \chi(p_3, \mathbf{W}(p_3, p_1))\]
or
\[\max_{p_3} \min_{p_2} \chi(p_3, \mathbf{W}(p_3, p_2))\text{ .}\]
This gives us  the upperbound.
\end{proof}

The same technique can be extended to a Byzantine multiple-access channel with \( k \) senders and at most one adversarial sender:

\begin{theorem}
Let \( \mathtt{W} \) be a classical-quantum \( k \)-user Byzantine multiple-access channel with at most one adversarial sender. Assume that the receiver decodes the messages in the order: sender 1, sender 2, ..., sender \( k \). For each sender \( i \in \{1, \dots, k\} \), we denote the arbitrarily varying classical-quantum channel that maps the message of sender \( i \) into the output state after having already measured the messages of senders \( 1, \dots, i-1 \) by \( \mathtt{W}_{i}^{j} \), where \( x^n_{j} \) and \( \gamma_j \) represent the jamming attack by the adversarial sender \( j \in \{1, \dots, i-1\} \).

The random capacity region of \( \mathtt{W} \) is 
\begin{align}
  R_i \leq \min\Bigl( \max_{p_i} \min_{p_{j}: j < i} I(p_i, E_j B | \{p_h : h \in \{1, \dots, i-1\} \setminus j\}), \\
  \max_{p_i} \min_{p_{j}: j > i} I(p_i, B | \{p_h : h \in \{1, \dots, i-1\}\}) \Bigr)
  \text{ .}
\end{align}
\label{lwback}
\end{theorem}

\begin{proof}
Similar to the proof of Theorem \ref{lwbactheo3}, we define 
\[
\Theta_i := \{ x^n_j : j \in \{1, \dots, k\} \setminus i \}
\]
and an arbitrarily varying classical-quantum channel
\[
\{\mathbf{W}(\cdot, t) : t \in \Theta_i\}: X_i \rightarrow X_1 \times \cdots \times X_{i-1} \otimes \mathcal{S}(A)
\]
denoted by
\begin{align*}
    &\mathbf{W}(x^n_i, x^n_j)\\ &= \sum_{m_1, \dots, m_k} p(x^n_1(m_1), \dots, x^n_{j-1}(m_{j-1}), x^n_{j+1}(m_{j+1}), \dots, x^n_{i-1}(m_{i-1}), x^n_{i+1}(m_{i+1}), \dots, x^n_k(m_k) | x^n_i) \\
    & \quad \cdot x^n_1(m_1) \otimes \cdots \otimes x^n_{i-1}(m_{i-1}) \\
    & \quad \otimes \sqrt{D_{m_j}} \mathtt{W}((x^n_1(m_1), \dots, x^n_{j-1}(m_{j-1}), x^n_j, x^n_{j+1}(m_{j+1}), \dots, x^n_{i-1}(m_{i-1}), x^n_i, x^n_{i+1}(m_{i+1}), \dots, x^n_k(m_k))) \sqrt{D_{m_j}}
\end{align*}
for \( j < i \), and
\begin{align*}
    \mathbf{W}(x^n_i, x^n_j) &= \sum_{m_1, \dots, m_k} p(x^n_1(m_1), \dots, x^n_{i-1}(m_{i-1}), x^n_{i+1}(m_{i+1}), \dots, x^n_{j-1}(m_{j-1}), x^n_{j+1}(m_{j+1}), \dots, x^n_k(m_k) | x^n_i) \\
    & \quad \cdot x^n_1(m_1) \otimes \cdots \otimes x^n_{i-1}(m_{i-1}) \\
    & \quad \otimes \mathtt{W}((x^n_1(m_1), \dots, x^n_{i-1}(m_{i-1}), x^n_i, x^n_{i+1}(m_{i+1}), \dots, x^n_{j-1}(m_{j-1}), x^n_j, x^n_{j+1}(m_{j+1}), \dots, x^n_k(m_k)))
\end{align*}
for \( j > i \).

When
\[
|M_3| < \min_{p_j: j < i} I(p_i, E_j B | \{ p_h : h \in \{1, \dots, i-1\} \setminus j \} )
\]
and
\[
|M_3| < \min_{p_j: j > i} I(p_i, B | \{ p_h : h \in \{1, \dots, i-1\} \} ),
\]
then by \cite{Ahl/Bli} and \cite{Da/Wa/Ha}, there exists a reliable random code for \( \{\mathbf{W}(\cdot, t) : t \in \Theta_i \} \).

Since the decoding operator can decode
the channel output of
$ \sum_{m_1,\cdots, m_k} p(x^n_{1}(m_1),\cdots,$ $x^n_{j-1}(m_{j-1}),x^n_{j+1}(m_{j+1}),\cdots,$ $x^n_{i-1}(m_{i-1}) ,x^n_{i+1}(m_{i+1}),\cdots,x^n_{k}(k)|x^n_{i})$
$ x^n_{1}(m_1)$ $\otimes\cdots
 \otimes$ $ x^n_{i-1}(m_{i-1})$
$\otimes \sqrt{D_{m_j}}$ $ \mathtt{W}((x^n_{1}(m_1),\cdots$ $,x^n_{j-1}(m_{j-1}),x^n_{j},x^n_{j+1}(m_{j+1}),$ $\cdots,x^n_{i-1}(m_{i-1}),x^n_{i},x^n_{i+1}(m_{i+1}),\cdots, x^n_{k}(m_{k})))$ $\sqrt{D_{m_j}}$,
every $D_{m_i}$ can be written as
\[D_{m_i}=\bigotimes_{m_1,\cdots,m_{i-1}}D_{m_i}^{(m_1,\cdots,m_{i-1})}\]
for some operators $\{D_{m_i}^{(m_1,\cdots,m_{i-1})}:m_1,\cdots,m_{i-1}\}$.

The receiver now applies a similar strategy to the quantum multiple access channel coding in \cite{Win2}. We assume that, after having already applied the measurements for the messages of senders 1 through \( i-1 \), the channel output results in
\[
\mathrm{tr} \Bigl( \mathtt{W}((x^n_1(m_1), \dots, x^n_{j-1}(m_{j-1}), x^n_j, x^n_{j+1}(m_{j+1}), \dots, x^n_k(m_k)) \prod_{h=1}^{i-1} D_{m_h} \Bigr)^{-1}
\]
\[
\sqrt{D_{m_1}} \cdot \cdots \cdot \sqrt{D_{m_{i-1}}} \mathtt{W}((x^n_1(m_1), \dots, x^n_{j-1}(m_{j-1}), x^n_j, x^n_{j+1}(m_{j+1}), \dots, x^n_k(m_k))) \sqrt{D_{m_{i-1}}} \cdot \cdots \cdot \sqrt{D_{m_1}}
\]
for some \( m_1, \dots, m_{i-1} \) with probability
\[
\mathrm{tr} \Bigl( \mathtt{W}((x^n_1(m_1), \dots, x^n_{j-1}(m_{j-1}), x^n_j, x^n_{j+1}(m_{j+1}), \dots, x^n_k(m_k)) \prod_{h=1}^{i-1} D_{m_h} \Bigr).
\]

Now the receiver uses
\[
\left\{ \mathrm{tr}\left( \sum_{m_i} D_{m_i}^{(m_1, \dots, m_{i-1})} \right)^{-1} D_{m_i}^{(m_1, \dots, m_{i-1})} : m_i \right\}
\]
to decode the message of sender \( i \). Let us first consider the adversarial sender \( j \) when \( j > i \). By the Tender Operator Lemma, the channel output
\[
\mathrm{tr} \left( \mathtt{W}((x^n_1(m_1), \dots, x^n_{j-1}(m_{j-1}), x^n_j, x^n_{j+1}(m_{j+1}), \dots, x^n_k(m_k)) \prod_{h=1}^{i-1} D_{m_h} \right)^{-1}
\]
\[
\sqrt{D_{m_1}} \cdot \dots \cdot \sqrt{D_{m_{i-1}}} \mathtt{W}((x^n_1(m_1), \dots, x^n_{j-1}(m_{j-1}), x^n_j, x^n_{j+1}(m_{j+1}), \dots, x^n_k(m_k)))
\]
\[
\sqrt{D_{m_{i-1}}} \cdot \dots \cdot \sqrt{D_{m_1}}
\]
after decoding the messages of senders 1 to \( i-1 \) is close to the original channel output 
\[
\mathtt{W}((x^n_1(m_1), \dots, x^n_{j-1}(m_{j-1}), x^n_j, x^n_{j+1}(m_{j+1}), \dots, x^n_k(m_k))).
\]

The expected value of a successful decoding is thus:
\begin{align*}
&\min_{x^n_j} \frac{1}{|\{m_i\}|} \sum_{m_1, \dots, m_k} p(x^n_1(m_1), \dots, x^n_k(m_k)) \\
&\quad \times \mathrm{tr} \left( \mathtt{W}((x^n_1(m_1), \dots, x^n_{j-1}(m_{j-1}), x^n_j, x^n_{j+1}(m_{j+1}), \dots, x^n_k(m_k)) D_{m_i}^{(m_1, \dots, m_{i-1})} \right) \\
&= \min_{x^n_j} \frac{1}{|\{m_i\}|} \sum_{m_i} p(x^n_i(m_i)) \\
&\quad \times \sum_{m_1, \dots, m_k} p(x^n_1(m_1), \dots, x^n_{i-1}(m_{i-1}), x^n_{i+1}(m_{i+1}), \dots, x^n_{j-1}(m_{j-1}), x^n_{j+1}(m_{j+1}), \dots, x^n_k(m_k) | x^n_i(m_i)) \\
&\quad \times \mathrm{tr} ( \Bigl( \sqrt{D_{m_1}} \cdot \dots \cdot \sqrt{D_{m_{i-1}}} \mathtt{W}((x^n_1(m_1), \dots, x^n_{j-1}(m_{j-1}), x^n_j, x^n_{j+1}(m_{j+1}), \dots, x^n_k(m_k))) \\
&\quad \left. \times \sqrt{D_{m_{i-1}}} \cdot \dots \cdot \sqrt{D_{m_1}} D_{m_i}^{(m_1, \dots, m_{i-1})} \Bigr) \right) - \epsilon \\
&= \min_{x^n_j} \frac{1}{|\{m_i\}|} \mathrm{tr} \left( \mathbf{W}(x^n_i(m_i), x^n_j) D_{m_i} \right) - \epsilon
\end{align*}
for a positive \( \epsilon \) that can be made arbitrarily close to zero for sufficiently large \( n \).
Similarly, the expected value of a successful decoding for the adversarial sender \( j \) when \( j < i \), for sufficiently large \( n \), is arbitrarily close to:
\begin{align*}
&\quad  \min_{x^n_j} \frac{1}{|\{m_i\}|} \sum_{m_1, \dots, m_k} p(x^n_1(m_1), \dots, x^n_{j-1}(m_{j-1}), x^n_{j+1}(m_{j+1}), \dots, x^n_k(m_k)) \\
&\quad \times \mathrm{tr} \Bigl( \sqrt{D_{m_j}} \mathtt{W}((x^n_1(m_1), \dots, x^n_{i-1}(m_{i-1}), x^n_i(m_i), x^n_{i+1}(m_{i+1}), \dots, x^n_{j-1}(m_{j-1}), x^n_j, x^n_{j+1}(m_{j+1}), \dots, x^n_k(m_k)))\\
&\quad \sqrt{D_{m_j}} D_{m_i}^{(m_1, \dots, m_{i-1})} \Bigr) \\
&= \min_{x^n_j} \frac{1}{|\{m_i\}|} \sum_{m_i} p(x^n_i(m_i)) \\
&\quad \times \sum_{m_1, \dots, m_k} p(x^n_1(m_1), \dots, x^n_{i-1}(m_{i-1}), x^n_{i+1}(m_{i+1}), \dots, x^n_{j-1}(m_{j-1}), x^n_{j+1}(m_{j+1}), \dots, x^n_k(m_k) | x^n_i(m_i)) \\
&\quad \times \mathrm{tr} ( \Bigl( x^n_1(m_1) \otimes \dots \otimes x^n_{i-1}(m_{i-1})\\ 
&\quad \sqrt{D_{m_j}} \mathtt{W}((x^n_1(m_1), \dots, x^n_{i-1}(m_{i-1}), x^n_i(m_i), x^n_{i+1}(m_{i+1}), \dots, x^n_{j-1}(m_{j-1}), x^n_j, x^n_{j+1}(m_{j+1}), \dots, x^n_k(m_k))) \\
&\quad \left. \times \sqrt{D_{m_j}} \Bigr) D_{m_i} \right) \\
&= \min_{x^n_j} \frac{1}{|\{m_i\}|} \mathrm{tr} \left( \mathbf{W}(x^n_i(m_i), x^n_j) D_{m_i} \right)
\end{align*}

This shows that the receiver can utilize the part of the code corresponding to \( m_1, \dots, m_{i-1} \) for the 
channel \( \{\mathbf{W}(\cdot, t) : t \in \Theta_{i}\} \) as an arbitrarily varying channel code in order to decode the message of sender \( i \).\vspace{0.2cm}

The converse can be shown in a similar as
the proof of the converse of Theorem
\ref{lwbactheo3}.
We consider now the case when sender i sends
$x^{\gamma_i}_i(m_i)$
through the channel.
Let us define
an arbitrarily varying classical-quantum channel
\[
\{\mathsf{W}(\cdot, t) : t \in \Theta_i\}: X_i \rightarrow   \mathcal{S}(A)
\]
denoted by
\begin{align*}
    &\mathsf{W}(x^n_i, x^n_j)\\ &= \sum_{m_1, \dots, m_k} p(x^n_1(m_1), \dots, x^n_{j-1}(m_{j-1}), x^n_{j+1}(m_{j+1}), \dots, x^n_{i-1}(m_{i-1}), x^n_{i+1}(m_{i+1}), \dots, x^n_k(m_k) | x^n_i) \\
    & \quad  \cdot  \sqrt{D_{m_j}} \mathtt{W}((x^n_1(m_1), \dots, x^n_{j-1}(m_{j-1}), x^n_j, x^n_{j+1}(m_{j+1}), \dots, x^n_{i-1}(m_{i-1}), x^n_i, x^n_{i+1}(m_{i+1}), \dots, x^n_k(m_k))) \sqrt{D_{m_j}}
\end{align*}
for \( j < i \), and
\begin{align*}
    \mathsf{W}(x^n_i, x^n_j) &= \sum_{m_1, \dots, m_k} p(x^n_1(m_1), \dots, x^n_{i-1}(m_{i-1}), x^n_{i+1}(m_{i+1}), \dots, x^n_{j-1}(m_{j-1}), x^n_{j+1}(m_{j+1}), \dots, x^n_k(m_k) | x^n_i) \\
    & \quad \cdot  \mathtt{W}((x^n_1(m_1), \dots, x^n_{i-1}(m_{i-1}), x^n_i, x^n_{i+1}(m_{i+1}), \dots, x^n_{j-1}(m_{j-1}), x^n_j, x^n_{j+1}(m_{j+1}), \dots, x^n_k(m_k)))
\end{align*}
for \( j > i \).
When sender j is 
adversarial and 
sends $x^n_j$ as jamming attack,
the expectetd value of 
the channel output will be
 closed to the arbitrarily varying classical-quantum channel
output 
$  \mathsf{W}(x^{\gamma_i}_i(m_i), x^n_j)$.
By \cite{Ahl/Bli} and \cite{Da/Wa/Ha}, 
the rate for the message of sender i
cannot exceed
\[ 
\min_{p_{j}} I(p_i, E_j B | \{p_h : h \in \{1, \dots, i-1\} \setminus j\})\]
when  $j$ is$<i$, and 
\[ \min_{p_{j}} I(p_i, B | \{p_h : h \in \{1, \dots, i-1\}\})\]
when $j$ is$<i$.
This gives us  the upperbound.
\end{proof}\vspace{0.2cm}

As stated in Section \ref{NoO}, the channel output can be altered by the decoding operator. In this section, we provide an example of a family of channels that are immune to adversarial decoding operators. We introduce a new symmetrizable condition as follows:

\begin{definition}\label{symmet2}
We say that the arbitrarily varying classical-quantum channel \( \{\mathbf{W}(\cdot, t) : t \in \theta\} \) is \emph{orthogonally symmetrizable} if there exists a parametrized set of distributions \( \{\tau(\cdot \mid x) : x \in X \} \) on \( \theta \) such that for all \( x, x' \in X \),
\[
\mathrm{tr} \left( \left( \sum_{t \in \theta} \tau(t \mid x) \mathbf{W}(x', t) \right) \left( \sum_{t \in \theta} \tau(t \mid x') \mathbf{W}(x, t) \right) \right) > 0.
\]
\end{definition}
It is clear that non-orthogonal symmetrizability is a very strong assumption, and very few channels will satisfy this condition. However, channels that are not orthogonally symmetrizable are still not affected by hostile decoding operators. We demonstrate this result for a $2$-user Byzantine multiple access channel, though this result can also be extended to a $k$-user Byzantine multiple access channel.

\begin{corollary}
Let $\mathtt{W}$ be a classical-quantum $2$-user Byzantine multiple access channel. Assume that the message of sender $1$ is to be decoded first, followed by the message of sender $2$. If $x^n_{1} \rightarrow \mathtt{W}((x^n_{1}, x^n_{2}))$ is not symmetrizable, and $x^n_{2} \rightarrow \mathtt{W}((x^n_{1}, x^n_{2}))$ is not orthogonally symmetrizable, then the capacity region of $\mathtt{W}$ is given by
\begin{align}
    R_1 &\leq \max_{p_1} \min_{p_2} I(p_1, B) \\
    R_2 &\leq \max_{p_2} \min_{p1} I(p_2, B) \text{ .} \label{lwbac2}
\end{align}

\label{lwbactheo2} \end{corollary}

\begin{proof}
The receiver first performs a measurement using an arbitrarily varying decoding operator set $\{D^{\gamma_1}_{m_1}: m_1 \in M_1\}$, which is constructed from the typical space projectors of the presumed code words of sender 1 for the arbitrarily varying channel $x^n_1 \rightarrow \mathtt{W}((x^n_1, x^n_2))$. If the size is smaller than $\min_{x^n_2} I(p_1, B_{x^n_2})$, the receiver can almost certainly decode the message of sender 1, even when sender 2 acts as a jammer.

Next, we consider the case where sender 1 is adversarial, sending a sequence $x_1^n$ as an attack, while sender 2 transmits $x_2(m_2)^{\gamma_2}$ for $m_2 \in M_2$. After the first measurement, the receiver then performs a second measurement using arbitrarily varying decoding operators $\{D^{\gamma_2}_{(m_1, m_2)}: m_1 \in M_1, m_2 \in M_2\}$ for the arbitrarily varying channel $(x^n_2) \rightarrow \mathtt{W}((x^n_1, x^n_2))$ with the knowledge of $m_1$ (cf. \cite{Win2}).

When $x^n_2 \rightarrow \mathtt{W}((x^n_1, x^n_2))$ is not orthogonally symmetrizable, we have that $\mathtt{W}((x^n_1, x^n_2(m_2)^{\gamma_2})) \perp \mathtt{W}((x^n_1(m_1), x^n_2(m_2')^{\gamma_2}))$ for all $m_2' \neq m_2$. Thus, we can construct a set $\{D^{\gamma_2}_{(m_1, m_2)}: m_2 \in M_2\}$ such that
\[
\mathrm{tr} \left( \mathtt{W}((x^n_1, x^n_2(m_2))) D_{m_1} \right)^{-1} \sqrt{D_{m_1}} \mathtt{W}((x^n_1, x^n_2(m_2))) \sqrt{D_{m_1}} D^{\gamma_2}_{(m_1, m_2')} < \epsilon
\]
for all $m_2' \neq m_2$. This ensures that the error probability of the code remains small after the first measurement on the channel output.

\end{proof}

\section{Conclusion and Future Directions}

In this work, we explored the impact of adversarial senders on Byzantine multiple-access classical-quantum channels. We demonstrated that a $2$-user Byzantine multiple-access channel, which may be a straightforward model in the classical setting, can exhibit non-trivial behavior in the quantum context. Specifically, we derived the capacity region for the $2$-user Byzantine multiple-access classical-quantum channels. Additionally, we established an achievable rate region for the $k$-user Byzantine multiple-access classical-quantum channels when $k \geq 3$.

An important challenge lies in applying advanced classical techniques, such as adversarial sender identification from \cite{Sa/Ba/De/Pr} and \cite{Sa/Ba/De/Pr2}, to quantum channels. The main difficulty, as mentioned earlier, is that for quantum arbitrarily varying channels, only random coding techniques are currently known.

\section*{Acknowledgments.}

The authors acknowledge the support of the German Federal Ministry of Education and Research (BMBF) under several national initiatives: the 6G Communication Systems initiative through the research hub 6G-life (grant 16KISK263), the QuaPhySI (Quantum Physical Layer Service Integration) initiative (grant 16KIS2234), the "QTOK - Quantum Tokens for Secure Authentication in Theory and Practice" initiative (grant 16KISQ038), the QR.N initiative (grant 16KIS2196), the QTREX initiative (grant 16KISR038), the QUIET initiative (grant 16KISQ0170), and the QD-CAMNetz initiative (grant 16ISQ169).


\begin{thebibliography}{xxx}
\bibitem{Ag/Ak} H. Aghaee and B. Akhbari, Classical-quantum mtiple access wiretap channel with common message: 
one-shot rate region, 2020 11th International Conference on Information and Knowledge Technology (IKT), 
55-61, Iran, 2020.
\bibitem{Ahl0} R. Ahlswede, A note on the existence of the weak capacity for channels with arbitrarily
varying channel probability functions and its relation to Shannon's zero error
capacity, Ann. Math.
Stat., Vol. 41, No. 3, 1970.
\bibitem{Ahl4} R. Ahlswede, Multi-way communication channels, Proceedings of 2nd International Symposium on Information Theory, Thakadsor, Armenian SSR, Sept. 1971, Akademiai Kiado, Budapest, 23-52, 1973. 
\bibitem{Ahl/Bj/Bo/No}R. Ahlswede,  I. Bjelakovi\'{c}, H. Boche,  and J. N\"otzel,
Quantum capacity under adversarial quantum noise:  arbitrarily
varying quantum channels, Comm. Math. Phys. A, Vol. 317, No. 1, 103-156,
2013.
\bibitem{Ahl/Bli} R. Ahlswede and  V. Blinovsky,  Classical capacity of classical-quantum arbitrarily
varying channels, IEEE Trans. Inform. Theory, Vol. 53, No. 2,
526-533, 2007.
\bibitem{Ahl/Ca} R.  Ahlswede and N. Cai, Arbitrarily varying multiple-access channels. I. Ericson's symmetrizability is adequate, Gubner's conjecture is true,IEEE Trans. Inform. Theory, Vol. 45, No. 2,  742-749, 1999.
\bibitem{Bj/Bo/Ja/No} I. Bjelakovi\'{c}, H. Boche, G. Jan\ss en, and J. N\"otzel,
Arbitrarily varying and compound classical-quantum channels and a
note on quantum zero-error capacities,  Information Theory, Combinatorics, and Search Theory, 
in Memory of
Rudolf Ahlswede, H. Aydinian, F. Cicalese, and C. Deppe eds., LNCS
Vol.7777,  247-283, 2012.
\bibitem{Bj/Bo/No} I. Bjelakovi\'{c}, H. Boche,  and J. N\"otzel,
Entanglement transmission and generation under channel uncertainty: universal quantum channel coding,
Commun. Math. Phys,
Vol. 292, No. 1, 55-97,  2009.
\bibitem{Bl/Br/Th2} D. Blackwell, L. Breiman, and A. J.
Thomasian, The capacities of a certain channel classes under random
coding, Ann. Math. Stat., Vol. 31, No. 3, 558-567, 1960.
\bibitem{Bo/Bo/Mo/fi} H. Boche, Y. N. B\"{o}ck, U. J. 
M\"{o}nich, and F, H. P. Fitzek, Trustworthy digital representations of analog
information - an application-guided analysis of
a fundamental theoretical problem in digital twinning,
Algorithms, Vol. 16, No. 11, 514, 2023.
\bibitem{Bo/De/No/Wi} H. Boche, C. Deppe, J. N\"otzel, and A. Winter,
Fully quantum arbitrarily varying channels: random coding capacity and capacity dichotomy,
2018 IEEE International Symposium on Information Theory (ISIT), 2012-2016,
 USA, 2018. 
\bibitem{Ch/Ne/Se} S. Chakraborty, A. Nema, and P. Sen, A multi-sender decoupling theorem and 
simultaneous decoding for the quantum MAC, 
2021 IEEE International Symposium on Information Theory (ISIT),   623-627, Australia, 2021.
\bibitem{Cs/Na} I.  Csisz\'ar and  P. Narayan, The capacity of the arbitrarily varying channel revisited: positivity, constraints,
IEEE Trans. Inform. Theory, Vol. 34, No. 2,  181-193, 1988.
\bibitem{Da/Wa/Ha}  A. Dasgupta, N. A. Warsi, and M. Hayashi,
 Universal tester for multiple independence testing and 
classical-quantum arbitrarily varying multiple access channel,
arXiv:2409.05619, 2024.
\bibitem{Fa/Ha/Sa/Se/Wi} 	O. Fawzi, P. Hayden, I. Savov, P. Sen, and M. M. Wilde, Classical communication over a quantum interference channel, 
 IEEE Trans.   Inform. Theory,  Vol. 58, No. 6, 3670-3691, 2012.
 \bibitem{Fitzek2021}
F. Fitzek and H. Boche, Research Landscape -- 6G Networks Research in Europe: 6G-life: Digital Transformation and Sovereignty of Future Communication Networks, IEEE Network, Vol. 35, No. 6, 4-5, Nov 2021.

\bibitem{Fitzek2022}
F. Fitzek, H. Boche, S. Stanczak, H. Gacanin, G. Fettweis, and H. Schotten,  6G Activities in Germany,
IEEE Future Networks Tech Focus, Vol. 15, Dec 2022.

\bibitem{Fettweis2022}
G. Fettweis and H. Boche,  On 6G and Trustworthiness, Communications of the ACM, Vol. 65, No. 4, 48-49, Apr 2022.

\bibitem{Gu} J. A. Gubner, On the deterministic-code capacity of the multiple-access arbitrarily varying channel, IEEE
Trans. Inform. Theory, Vol. 36, No. 2,  262-275, 1990. 
\bibitem{Hay} M. Hayashi, Universal coding for classical-quantum channel,
Commun. Math. Phys, Vol. 289, No 3, 1087-1098, 2009.
\bibitem{Hay2}  M. Hayashi, Quantum Information Theory, Springer-Verlag Berlin Heidelberg, 2017.
\bibitem{Ha/Na} M. Hayashi, H. Nagaoka, General formulas for capacity of 
classical-quantum channels, IEEE Trans. Inform. Theory, Vol. 49. No. 7, 1753-1768,
2003. 
\bibitem{Ja/La/Ka/Ho/Ka/Me} S. Jaggi, M. Langberg, S. Katti, T. Ho, D. Katabi, and M. M\'{e}dard, Resilient network coding in the presence
of byzantine adversaries, in IEEE INFOCOM 2007-26th IEEE International Conference on Computer
Communications,  616-624, IEEE, 2007.
\bibitem{Ja} J.-H. Jahn, Coding of arbitrarily varying multiuser channels, IEEE Trans. Inform. Theory,
Vol. 27, No. 2, 212-226, 1981.
\bibitem{Ko/To/Da}  O. Kosut, L. Tong, and N. David, Polytope codes against adversaries in networks, IEEE Trans. 
Inform. Theory, Vol. 60, No. 6,  3308-3344, 2014.
\bibitem{Lia} H. Liao, Multiple access channels, Ph.D. dissertation, Department of Electrical Engineering, University of Hawaii,
Honolulu, 1972.
\bibitem{Mo} M. Mosonyi,
Coding theorems for compound problems via quantum R\'{e}nyi divergences,
 IEEE Trans. Inform. Theory, Vol. 61, No. 6, 2997-3012, 2015.
\bibitem{Og/Na} T. Ogawa and H. Nagaoka, Making good codes for
classical-quantum channel coding via quantum hypothesis testing,
IEEE Trans. Inf. Theory, Vol. 53, No. 6, 2261-2266, 2007.
\bibitem{Sa/Ba/De/Pr}  N. Sangwan, M. Bakshi, B. K. Dey, 
and V. M. Prabhakaran, Byzantine Multiple Access Channels - Part 
I: Reliable Communication, IEEE Trans.   Inform. Theory, 
Vol.  70, No. 4, 2309-2366, 2024.
\bibitem{Sa/Ba/De/Pr2} N. Sangwan, M. Bakshi, B. K. Dey, 
and V. M. Prabhakaran, Byzantine Multiple Access Channels - 
Part II: Communication With Adversary Identification,IEEE Transactions on Information Theory, vol. 71, no. 1, pp. 23-60, Jan. 2025.
\bibitem{Se} P.  Sen, Achieving the Han-Kobayashi inner bound for
the quantum interference channel,  2012 IEEE International Symposium 
on Information Theory (ISIT), 736-740, USA, 2012. 
\bibitem{Wa/Si/Ks} D. Wang, D. Silva, and F. R. Kschischang, Robust network coding in the presence of untrusted nodes, IEEE
Trans. Infor. Theory, Vol. 56, No. 9,  4532-4538, 2010.
\bibitem{Win} A. Winter, Coding theorem and strong converse for quantum
channels, IEEE Trans. Inf. Theory, Vol. 45, No. 7,  2481-2485,
1999.
\bibitem{Win2}  A. Winter, The capacity of the quantum multiple-access channel,
IEEE Trans.   Inform. Theory, vol. 47, no. 7,  3059-3065, 2001.
\bibitem{Ya/Ha/De2}J. Yard, P. Hayden, and I. Devetak, Capacity theorems for quantum multiple
access channels --- classical-quantum and quantum-quantum capacity
regions, IEEE Trans. Inf. Theory, Vol. 54, No. 7, 3091-3113,
2008.

\end{thebibliography}
\end{document}